\newif\ifprivate
\privatetrue

\documentclass[a4paper,10pt]{article}

\usepackage[utf8]{inputenc}
\usepackage{amsmath,amssymb}
\usepackage{amsthm}
\usepackage{a4wide}
\usepackage{comment}
\usepackage{hyphenat}
\usepackage{paralist}
\usepackage[numbers,square]{natbib}
\usepackage{xcolor}
\usepackage{hyperref}
\usepackage[sort&compress,capitalize,nameinlink]{cleveref}

\usepackage{authblk}

\usepackage{tikz}
\usetikzlibrary{calc}

\newtheorem{theorem}{Theorem}

\newtheorem{corollary}{Corollary}
\newtheorem{lemma}{Lemma}

\newtheorem{claim}{Claim}

\newtheorem{observation}{Observation}
\theoremstyle{definition}

\newtheorem{definition}{Definition}
\theoremstyle{remark}
\newtheorem{remark}{Remark}

\newcommand{\inst}{\ensuremath{I}}

\newcommand{\yes}{``Yes''}
\newcommand{\no}{``No''}

\newcommand{\N}{\mathbb{N}}

\newcommand{\cocl}[1]{\ensuremath{\operatorname{#1}}}

\newcommand{\NP}{\cocl{NP}}

\newcommand{\coNP}{\cocl{coNP}}

\newcommand{\poly}{\cocl{poly}}

\newcommand{\NPincoNPslashpoly}{\ensuremath{\NP\subseteq \coNP/\poly}}

\renewcommand{\O}{\ensuremath{{O}}}

\newcommand{\vareps}{\varepsilon}

\newcommand{\dgn}{\ensuremath{p}}

\newcommand{\oct}{\ensuremath{p}}

\newcommand{\prob}[1]{\textsc{#1}}

\newcommand{\Clique}{\prob{Clique}}

\newcommand{\ceil}[1]{\lceil#1\rceil}

\newcommand{\TK}{Turing kernel}

\newcommand{\ORK}{OR kernel}
\newcommand{\xbar}[1]{\overline{#1}}

\usepackage{xargs}
\newcommandx{\set}[2][1=1]{\ensuremath{\{#1,\dots,#2\}}}
\newcommandx{\tlog}[3][1=,3=]{\ensuremath{\log_{#1}^{#3}(#2)}}

\usepackage{mathtools}

\newcommand{\ceq}{\coloneqq}

\usepackage{etoolbox}

\ifprivate{}
    \usepackage[textsize=scriptsize,color=green!20!white]{todonotes}
\else{}
    \usepackage[disable,textsize=scriptsize]{todonotes}
\fi{}

\newcommand{\tikzinst}[4]{
  \node (#1) at (#2,#3)[rounded corners,fill=lightgray!33!white,minimum width=\xr*1.25cm,minimum height=\yr*0.75cm,draw,dashed,font=\small]{#4};
}
\newcommandx{\tikzpinst}[5][5=2.5]{
  \node (#1) at (#2,#3)[rounded corners,fill=lightgray!33!white,minimum width=\xr*#5 cm,minimum height=\yr*0.75cm,draw,font=\small]{#4};
}

\title{Polynomial Turing Kernels for Clique\\ with an Optimal Number of Queries}
\author[1]{Till Fluschnik\thanks{Supported by DFG, projects TORE (NI 369/18) and MATE (NI 369/17).}}
\author[1]{Klaus Heeger\thanks{Supported by DFG RGT 2434 ``Facets of Complexity''.}}
\author[2]{Danny Hermelin}
\affil[1]{Technische Universit\"at Berlin, Algorithmics and Computational Complexity}
\affil[2]{Ben-Gurion University of the Negev, Industrial Engineering and Management}

\begin{document}

\maketitle
\sloppy

\begin{abstract}
A polynomial Turing kernel for some parameterized problem $P$ is a polynomial-time algorithm that solves $P$ using queries to an oracle of $P$ whose sizes are upper-bounded by some polynomial in the parameter. Here the term ``polynomial" refers to the bound on the query sizes, as the running time of any kernel is required to be polynomial. One of the most important open goals in parameterized complexity is to understand the applicability and limitations of polynomial Turing Kernels. As any fixed-parameter tractable problem admits a Turing kernel of some size, the focus has mostly being on determining which problems admit such kernels whose query sizes can be indeed bounded by some polynomial.  

In this paper we take a different approach, and instead focus on the number of queries that a Turing kernel uses, assuming it is restricted to using only polynomial sized queries. Our study focuses on one the main problems studied in parameterized complexity, the \Clique{} problem: Given a graph $G$ and an integer $k$, determine whether there are $k$ pairwise adjacent vertices in $G$. We show that \Clique{} parameterized by several structural parameters exhibits the following phenomena: 
\begin{itemize}
\item It admits polynomial Turing kernels which use a sublinear number of queries, namely $O(n/ \tlog{n}[c])$ queries where $n$ is the total size of the graph and $c$ is any constant. This holds even for a very restrictive type of Turing kernels which we call OR-kernels.
\item It does not admit polynomial Turing kernels which use $O(n^{1-\varepsilon})$ queries, unless $\NP \subseteq \coNP/\poly$.
\end{itemize}
For proving the second item above, we develop a new framework for bounding the number of queries needed by polynomial Turing kernels. This framework is inspired by the standard lower bounds framework for Karp kernels, and while it is quite similar, it still requires some novel ideas to allow its extension to the Turing setting. 
\end{abstract}

\section{Introduction}

One of the most important goals in parameterized complexity research these days is to understand the applicability and limitations of Turing kernels. A kernel, in rough terms, is a polynomial-time reduction from a problem onto itself where the size of the output is bounded by a function in the parameter of the input (called the \emph{size} of the kernel).
Generally speaking, the reduction can be either of Karp-type or Turing-type, where in the latter case one uses the term \emph{Turing kernel} to emphasize the use of the weaker notion of reduction.
Thus, a Turing kernel for a parameterized problem $\Pi$ is an algorithm that solves $\Pi$ in polynomial time using queries to $\Pi$ of size bounded in the parameter of the input instance.
\begin{definition}[Turing kernel]
A \emph{Turing kernel} for a parameterized problem $L \subseteq \{0,1\}^* \times \mathbb{N}$ is an algorithm that has access to an oracle of $L$, such that on input $(x,p) \in \{0,1\}^* \times \N$, the algorithm correctly determines in polynomial-time whether $(x,p) \in L$ using only queries ``$(y,p') \in L$?'' with $|y|+p' \leq f(p)$, for some computable function $f$. The function $f$ is referred to as the \emph{size} of the kernel. 
\end{definition}

Not only is kernelization one of the most powerful techniques in parameterized complexity, 
it also provides an alternative way of formalizing the central concept of tractability in the area~\cite{DowneyF13}. For this reason, kernelization has received wide attention and has been the focus of extensive research~\cite{Kratsch14,LokshtanovMS12}. 
However, much more is known on Karp kernels. 
In particular, 
the recent book on kernelization~\cite{FominLSZ2019} is mostly devoted to techniques for constructing Karp kernels of polynomial size. 
Moreover, 
there is also a widely applicable framework for showing that certain problems do not admit Karp kernels of any polynomial size~\cite{BodlaenderDFH09}, 
and even of a specific polynomial size~\cite{DellM14}, 
assuming the widely believed assumption that $\coNP\nsubseteq\NP/\poly$. 
Both of these fact combined made Karp kernels of polynomial size the golden standard in kernelization. 

In contrast, the situation is quite different for Turing kernels. 
First off, there are only a handful of known problems which admit polynomial-size Turing kernels (yet do not admit polynomial-size Karp kernels). 
Notable examples are \textsc{Max Leaf Subtree}~\cite{BFFLSV12} and \textsc{Planar Longest Cycle}~\cite{Jansen17} under the standard parameterizations of these problems, 
and \prob{Clique} parameterized by the vertex cover number~\cite{BodlaenderJK14}. 
On the other hand, 
there is only a rather limited framework for excluding Turing kernels of polynomial size which is based on non-standard complexity-theoretic assumptions~\cite{HermelinKSWW15} or diagonalization~\cite{WitteveenBT19}.
This might be attributed to the fact that $\coNP=\NP$ (and so is a subset of $\NP/\poly$) under Turing reductions, 
and so the standard implication for excluding Karp kernels might not apply here. 
Thus, one of the central open problems in parameterized complexity is to establish a framework for excluding Turing kernels of polynomial size under standard complexity assumptions.  

\subsection{An \ORK{} for Clique}

To illustrate the discrepancy between Karp and Turing kernels let us look at the classical \prob{Clique} problem: 
Given a graph $G$ on $n_G$ vertices, 
and a positive integer $k$,
determine whether there are $k$ pairwise adjacent vertices in $G$. 
Consider this problem when the parameter is $\Delta$, the maximum degree of $G$. Then it is trivial to show that this problem has no polynomial Karp kernel using the standard lower bound framework~\cite{BodlaenderDFH09}. However, \prob{Clique} parameterized by $\Delta$ admits an easy and practical Turing kernel of polynomial size: Simply query the closed neighborhood of each vertex in the input graph to search for the desired clique.
This kernel has $n_G$ queries, 
and each query is of size~$O(\Delta^2)$ (here we count the total number of vertices and edges in the query).

The example above is actually a very restricted version of a Turing kernel which we refer to in this paper as an \emph{\ORK{}}. In particular, this algorithm is \emph{non-adaptive} in the sense that each query is independent of the answers to the previous queries. Thus, the queries can be constructed in advance without prior knowledge to the answer of each query. We use the term ``\ORK{}" to emphasize the fact that the algorithm functions similar to a logical OR gate in that the algorithm answers ``Yes'' if and only if the answer to any one of the queries is ``Yes'' (the entire graph has a clique of size $k$ if and only if one of the closed neighborhoods in the graph has a clique of size $k$). 
\begin{definition}[OR-kernel]
\label{def:ORkernel}
An \emph{\ORK{}} for a parameterized problem $L \subseteq \{0,1\}^* \times \mathbb{N}$ is an algorithm that on a given  $(x,p) \in \{0,1\}^* \times \mathbb{N}$, outputs in polynomial time a sequence of instances $(y_1,p_1),\ldots,(y_t,p_t) \in \{0,1\}^* \times \mathbb{N}$ such that:
\begin{itemize}
\item $\max_{i\in\set{t}} |y_i|+p_i \leq f(p)$ for some function $f$, and  
\item $(x,p) \in L$ if and only if $(y_i,p_i) \in L$ for some $i\in\set{t}$. 
\end{itemize}
The function $f$ is referred to as the \emph{size} of the kernel, and $t=t(n,p)$ is the number of queries used by the kernel.
\end{definition}

\ORK{}s have various advantages over Turing kernels. First of all, their non-adaptive nature allows to process the queries in parallel or in a distributed fashion. This is an even more crucial aspect when considering the prevalence of cloud computation these days.
In such settings, it makes sense to not only optimize the size $f$ of the desired kernel, as is done when considering Karp kernels, but also the number of queries $t$, as
there is typically a cost involved for running extra machines (physical or virtual). This is also to true, albeit to a lesser extent, with general Turing kernels. 

\subsection{Our contribution}

In this paper we abandon the pursuit of developing a framework for the exclusion of polynomial-size Turing kernels under standard complexity assumptions. Instead, we focus on developing lower bounds for the number of queries used by such kernels; that is, Turing kernels that are restricted to have only polynomial-size queries. Thus, going back to the example above, we wish to answer questions of the following type:
\begin{quote}
\emph{Does \prob{Clique} parameterized by the maximum degree $\Delta$ admit a polynomial-size \TK{} with $o(n)$ queries?}    
\end{quote}

It turns out that the answer to this question is positive, even when restricted to OR-kernels.
In particular, the \prob{Clique} problem has specific properties that allow us to employ a 
a novel \emph{query batching} technique, where we query several instances together as a single query. This technique allows us to shave off a polylogarithmic factor from the total number of queries for several \ORK{}s for \prob{Clique}, under various parameterizations (see also \cref{fig:paramhier}). 
\begin{theorem}
\label{thm:UpperBound}
For every~$c \in \mathbb{R}$, \textsc{Clique} admits a polynomial \ORK{} with $O(n/\tlog{n}[c])$ queries\footnote{Throughout the paper, we use $n$ to denote the total encoding length of the input graph, and not just its number of vertices; see Section~\ref{sec:pre} for the precise details.}, for all the following parameters: 
Bandwidth, 
cutdwidth, 
degeneracy, 
distance to bounded degree, 
edge bipartization, 
feedback vertex set, 
$k$ + distance to chordal, 
$k$ + distance to interval, 
maximum degree, 
odd cycle traversal, 
treedepth,
vertex cover number,
length of the longest odd cycle,
pathwidth, and 
treewidth.   
\end{theorem}

While shaving off polylogarithmic factors from the total number of queries is a substantial improvement, one might naturally ask whether we could do even better. For this, we develop a novel framework for excluding polynomial-size \TK{}s with with $O(n^{1-\vareps})$ queries, which is particularly useful for the \prob{Clique} problem. Our framework is based on the standard framework for excluding polynomial Karp kernels~\cite{BodlaenderDFH09}, and in particular on a refinement of the Fortnow and Santhanam argument~\cite{FortnowS11} that lies in the heart of that framework. Using our new framework, we rule out such \TK{}s for \prob{Clique} under all parameters listed in \cref{thm:UpperBound}, showing that the kernels stated in the theorem are essentially optimal in terms of the number of queries.  
\begin{theorem}
\label{thm:LowerBound}
Unless~\NPincoNPslashpoly, the \prob{Clique} problem admits no polynomial-size Turing kernel with $O(n^{1-\vareps})$ queries, for any $\vareps > 0$, for all the parameters listed in \cref{thm:UpperBound}. 
\end{theorem}

The proof for \cref{thm:LowerBound} is given in \cref{sec:lower-bounds}, while the proof of \cref{thm:UpperBound} is given in \cref{sec:upper-bounds}. Together, both theorems give quite a clear picture for the number of queries any polynomial-size Turing kernel \prob{Clique} can have, under several natural parameterizations. All these parameters, and the relations between them, are depicted in Figure~\ref{fig:paramhier}.

\begin{figure}[h!]
 \centering
 \begin{tikzpicture}
  \def\xr{1.25}
  \def\yr{0.45}
  \newcommand{\parabox}[5]{
    \node (#1) at (#2*\xr,#3*\yr)[minimum width=1.5*\xr cm, minimum height=1*\yr cm,text width=1.5*\xr cm, rounded corners, align=center,draw,fill=gray!5!white,font=\scriptsize,#5]{#4};
  }
  
  \parabox{deg}{2}{0}{degeneracy}{very thick,draw=green!50!black}
  \parabox{delta}{2}{2}{maximum degree}{}
  \parabox{bw}{2}{6}{Bandwidth}{very thick,draw=red!50!black}
  \parabox{cw}{3}{4}{Cutwidth}{very thick,draw=red!50!black}

  \parabox{tw}{0}{2}{Treewidth}{}
  \parabox{pw}{0}{4}{Pathwidth}{}
  \parabox{fvs}{-2}{4.5}{Feedback vertex set number}{}
  \parabox{dti}{-4}{4}{Distance to Interval~$+~k$}{}
  \parabox{ebi}{-6}{4}{Edge bipartization}{}
  \parabox{dtc}{-2}{1.5}{Distance to Chordal~$+k~$}{very thick,draw=green!50!black}
  \parabox{oct}{-4}{2}{Odd cycle transversal}{very thick,draw=green!50!black}
  \parabox{dbd}{-6}{6.25}{Distance to Bounded Degree}{very thick,draw=green!50!black}
  \parabox{td}{0}{6}{Treedepth}{}
  \parabox{vc}{-1}{8}{Vertex cover number}{very thick,draw=red!50!black}
  
  \parabox{loc}{4}{6}{Longest odd cycle length}{very thick,draw=green!50!black}
  
  \foreach\x/\y in{td/vc,tw/pw,tw/fvs,fvs/vc,pw/td,pw/bw,deg/tw,deg/delta,delta/bw,delta/cw,dtc/dti,dbd/vc,oct/fvs,oct/ebi,ebi/vc,loc.north/vc}{\draw[->,>=latex,thick] (\x) to (\y);}
  \draw[->,>=latex,thick,dashed] (dti) to (vc);
  \draw[->,>=latex,thick,dashed] (dtc) to (fvs);

 \end{tikzpicture}
\caption{Parameters investigated in this paper. An arrow from some parameter~$p_1$ to another parameter~$p_2$ means that $p_1 = \poly(p_2)$ (dashed arrows mean that $p_1 = \poly (p_2)$ on all Yes-instances). 
Thus, 
for showing upper bounds it suffices to focus on the source parameters of this dag (thick green frames), and for lower bounds it is enough to focus on the sink parameters (red frames).}
 \label{fig:paramhier}
\end{figure}
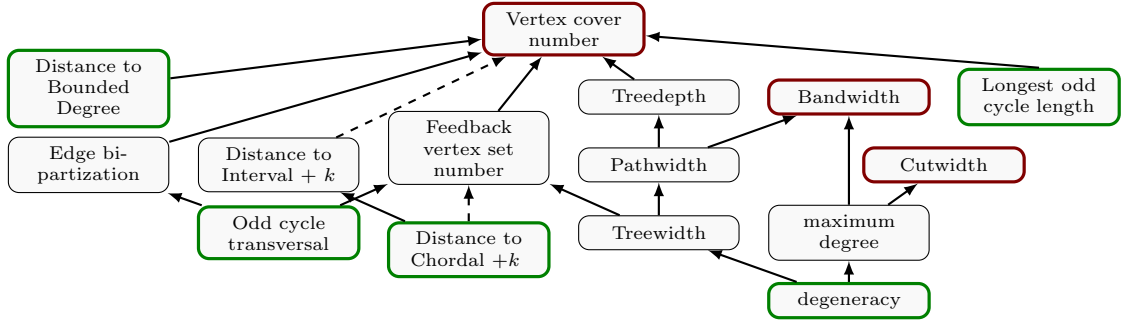

\section{Preliminaries}
\label{sec:pre}

We use basic notation from graph theory~\cite{Diestel10} and parameterized complexity~\cite{CyganFKLMPPS15}. Throughout the paper, we will work on the binary alphabet $\Sigma := \{0, 1\}$. For a non-negative integer~$x$, we let~$\{0,1\}^{\leq x}$ denote the set of strings of length at most $x$ over the alphabet~$\{0,1\}$. We also use $\{0,1\}^*$ to denote $\bigcup_{x \in \N} \{0,1\}^{\leq x}$. A \emph{classical problem} is a subset of~$\{0,1\}^*$. A \emph{parameterized problem} is a subset of~$\Sigma^* \times \mathbb{N}$.
For an instance~$(y, p)$ of a parameterized problem, we define~$|(y, p)| := |y| + p + c$, where $|y|$ is the encoding length of~$y$ and $c$ is a constant sufficiently large to encode the parentheses. For sequence~$(y_1, p_1), \dots, (y_k, p_k)$ of instances of a parameterized problem we define $|(y_1, p_1), \dots, (y_k, p_k)| : = \sum_{i =1}^k |(y_i, p_i)|$.

The main protagonist of this paper is the \Clique{} problem: Given a graph~$G$ and an integer~$k$, determine whether there are $k$ pairwise adjacent vertices in~$G$ (\emph{i.e.}, whether $G$ contains a $k$-clique). For a graph~$G$, we denote by~$n_G$ and~$m_G$ the number of vertices and edges of~$G$, respectively. We assume throughout that graphs are represented using adjacency lists. In this way we need $O(n_G +m_G \tlog{n_G})$ bits to encode $G$, and so the encoding length~$n$ of an instance $(G,k)$ of \Clique{} is $O(n_G +m_G \tlog{n_G} + \lg(k))=O(n_G +m_G \tlog{n_G})$. We refer to $n$ as the \emph{size} of $(G,k)$ (or just $G$ for simplicity). Finally, we use $\poly(n)$ instead of $n^{\O(1)}$ when convenient.

\section{Lower Bounds}
\label{sec:lower-bounds}

In the following we introduce a framework to exclude polynomial-sized Turing kernel with $O (n^{1-\vareps})$-queries under the standard assumption of $\NP \not \subseteq \coNP/\poly$, and in particular, we provide a complete proof for \cref{thm:LowerBound}. The basic idea for excluding such Turing kernels is the same as the one for excluding polynomial-sized Karp kernels~\cite{BodlaenderDFH09}: Combining a (linear) OR composition (see \Cref{def:lin-comp}) with a polynomial-sized Turing kernel that uses only $O(n^{1-\vareps})$-queries results in an (oracle-OR-) distillation (see \Cref{def:oracleOR}); this is shown in \Cref{thm:CompositionTuringKernel}. The latter can then be used to compress $t$ instances of an NP-hard problem into queries with total size at most~$t-1$, which in turn can be used to construct a polynomial-size advice for the NP-hard problem, resulting in $\NP \subseteq \coNP/\poly$ (see \Cref{thm:oracle OR distillation}). All together, this will give us the following theorem:
\begin{theorem}
\label{thm:lower-bound-main}
Let $R$ be any parameterized problem. Suppose that there is a linear OR composition from some \NP-hard problem~$L$ to $R$. Then~$R$ admits no polynomial-size Turing kernel with $O(n^{1-\vareps})$ queries, for any $\vareps > 0$, unless $\NP \subseteq \coNP/\poly$. 
\end{theorem}

\subsection{Generic Framework}

We start by defining oracle OR distillations, which are basically algorithms that have access to an oracle of a problem $R$, and solve the OR of $t$ instances of a problem $L$ in polynomial time. More precisely: 
\begin{definition}[Oracle OR distillation]
\label{def:oracleOR}
Let $L$ be a classical problem, and let $R$ be a parameterized problem. An \emph{oracle OR distillation} of~$L$ into~$R$ is an algorithm~$\mathcal{A}$ with oracle access to~$R$ such that there exists some $\vareps  > 0$, where given any sequence of instances $x_1, \dots, x_t$ for~$L$, the algorithm~$\mathcal{A}$
\begin{itemize}
\item uses queries $(y_1,p_1),\ldots,(y_s,p_s)$ to the oracle of $R$ whose total size $\sum_{i=1}^s |(y_i,p_i)|$ is bounded by $t^{1-\vareps} \poly (\max_{i \in \{1, \dots, t\}} |x_i|)$, and
\item returns \yes\ if and only if there exists at least one index~$i$ such that $x_i\in L$.
\end{itemize}
\end{definition}

We now show how an OR composition together with a polynomial-size Turing kernel with $O(n^{1-\vareps})$-queries can be used to construct an oracle OR distillation. However, for this to work, the size of the output of the OR composition can only grow linearly in the number~$t$ of instances it composes. That is, the output should have size at most $t \poly (n)$, where $n$ is the size of the largest input instance. We refer to OR compositions fulfilling this additional property as \emph{linear} OR compositions. Formally:
\begin{definition}
\label{def:lin-comp}
Let~$L$ be a classical problem, and let~$R$ be a parameterized problem.  A \emph{linear OR composition} from~$L$ to~$R$ is an algorithm that takes~$t$ instances~$x_1,\dots,x_t$ of~$L$ and computes in time~$\poly(\sum_t |x_t|)$ an instance~$(y,p)$ of~$R$ such that
\begin{enumerate}[(i)]
\item $|(y,p)| = t \cdot \poly (\log t) \cdot \poly(\max_t |x_t|)$,
\item $p= \poly(\max_t |x_t| + \tlog{t})$, and
\item $(y,p)\in R$ if and only if~$x_i\in L$ for some~$i\in\set{t}$.
\end{enumerate}
\end{definition}

\begin{theorem}
\label{thm:CompositionTuringKernel}
Let $L$ be an \NP-hard problem, and let $R$ be a parameterized problem. Suppose that there is a linear OR composition from~$L$ to $R$, and that $R$ admits a polynomial-size Turing kernel with $O(n^{1-\vareps})$ many queries for some $\vareps > 0$. Then $L$ admits an oracle OR distillation.
\end{theorem}

\begin{proof}
Let $x_1, \dots, x_t$ be instances of $L$, and let $n \ceq \max_{i\in\set{t}} |x_i|$. We assume without loss of generality that $x_i \neq x_j$ for every $i \neq j$ (identical instances can safely be removed). Thus, we have $t \le 2^n$, and consequently $\log(t)=O(n)$. We construct an oracle OR distillation for $L$ as follows (see \Cref{fig:my_label}).
\begin{enumerate}
\item Apply the linear OR composition to $x_1, \dots, x_t$, resulting in an instance $(y,p)$ of $R$ with $(y,p) \in R$ if and only if there is some $i\in \{1  , \dots, t\}$ such that $x_i \in L$. By the requirements of a linear OR composition we have $|(y,p)| = t \cdot \poly ( \log t) \cdot \poly(n) = t \cdot \poly (n)$, and~$p = \poly (n+\log(t)) = \poly(n)$, since $\log(t)=O(n)$.
\item Apply the assumed polynomial Turing kernel to~$(y, p)$. This Turing kernel queries an oracle for~$R$ with instances $(y_1,p_1), \ldots, (y_s,p_s)$ such that $|(y,p_i)|=\poly(p)$ for each $i \in \{1,\ldots,s\}$, and $s = (t\cdot \poly(n))^{1-\vareps} = t^{1-\vareps} \cdot \poly(n)$. Observe that we have 
\[
\sum_{i=1}^s |(y_i,p_i)| = t^{1-\vareps} \cdot \poly(n) \cdot \poly (p) = t^{1-\vareps} \cdot\poly(n),
\]
since $p=\poly(n)$. 
\item Return \yes{} if and only if the Turing kernel returns \yes{}.
\end{enumerate}

We argue that the algorithm above is indeed a oracle OR distillation. The oracle OR distillation calls the oracle only for the queries $y_1, \dots, y_s$, and the total size of these queries is bounded by $t^{1-\vareps}\poly (n)$. Furthermore, the oracle OR distillation returns \yes{} if and only if there exists some~$i\in\set{t}$ such that $x_i \in L$. Thus, the algorithm above satisfies both requirements of \Cref{def:oracleOR}, and the lemma follows.
\end{proof}

\begin{figure}[hbt]
\centering
\begin{tikzpicture}
\def\xr{1}
\def\yr{1}

\tikzpinst{cc}{0}{-2*\yr}{$(y,p)$};
        
\foreach \x in {1,...,6}{
\if\x5 
\node (x\x) at (1.5*\x*\xr-3.5*1.5*\xr,0*\yr)[minimum height=\yr*0.75cm]{$\cdots$};
\else
\if\x6 
\tikzinst{x\x}{1.5*\x*\xr-3.5*1.5*\xr}{0*\yr}{$x_{t}$};
\else
\tikzinst{x\x}{1.5*\x*\xr-3.5*1.5*\xr}{0*\yr}{$x_{\x}$};
\fi
\fi
\draw[->,thick,>=latex] (x\x) to [out=-82-\x,in=90](cc);
}
        
\foreach \x in {1,...,4}{
\if\x3
\node (tk\x) at (2.25*2.5*\xr-2*2.5*2.25*\xr+2.25*\x*\xr,-2*\yr-1.5*\yr)[minimum height=\yr*0.75cm]{$\cdots$};
\else
\if\x4
\tikzpinst{tk\x}{2.25*2.5*\xr-2*2.5*2.25*\xr+2.25*\x*\xr}{-2*\yr-1.5*\yr}{$(y_{s},p_s)$}[1.75];
\else
\tikzpinst{tk\x}{2.25*2.5*\xr-2*2.5*2.25*\xr+2.25*\x*\xr}{-2*\yr-1.5*\yr}{$(y_{\x},p_\x)$}[1.75];
\fi
\fi
}
        
\node (tk) at (2.25*2.5*\xr-2.5*2.25*\xr,-2*\yr-1.5*\yr)[minimum width=8.5*\xr cm,minimum height=1.25*\yr cm,dashed,thick,rounded corners,draw]{};
        
\draw[->,thick,>=latex] (cc) to (tk);
\draw[->,thick,>=latex] (tk1) to (tk2);
\draw[->,thick,>=latex] (tk2) to (tk3);
\draw[->,thick,>=latex] (tk3) to (tk4);
        
\node (output) at (2.25*2.5*\xr-2.5*2.25*\xr, -2*\yr-3*\yr)[fill=lightgray!33!white,minimum width=\xr cm,minimum height=\yr*0.75cm,draw,font=\small]{\yes/\no};
        
\draw[->,thick,>=latex] (tk) to [out=-90,in=90](output);
        
\def\deps{0.01}
\def\xsh{0.25}

\node[left =\xsh*1 cm of x1,align=right,font=\small]{instances\\ of~$L$};
\node[left =\xsh*1 cm of x1,yshift=-1*\yr cm,align=right,font=\small]{linear OR-\\ composition};
\node[left =\xsh*1 cm of x1,yshift=-2*\yr cm,align=right,font=\small]{instance\\ of~$R$};
\node[left =\xsh*1 cm of x1,yshift=-3*\yr cm,align=right,font=\small]{Turing kernel};
\node[left =\xsh*1 cm of x1,yshift=-3.5*\yr cm,align=right,font=\small]{queries\\ of~$R$};
\node[left =\xsh*1 cm of x1,yshift=-5*\yr cm,align=right,font=\small]{Output};

\node[right =\xsh*1 cm of x6,font=\small]{$|x_i|\leq n$};
\node[right =\xsh*1 cm of x6,yshift=-3.5*\yr cm,align=left,font=\small]{$\sum_{i=1}^s |(y_i, p_i)|=$\\$ t^{1-\vareps}\cdot \poly(n)$};
\node[right =\xsh*1 cm of x6,yshift=-2*\yr cm,align=left,font=\small]{$|(y,p)|\leq t\cdot \poly(n)$};
\end{tikzpicture}
\caption{An oracle OR distillation for $R$ obtained by combining a linear OR composition of~$L$ to~$R$, and a polynomial-size Turing kernel for $R$ with $O(n^{1-\vareps})$ queries.}
\label{fig:my_label}
\end{figure}
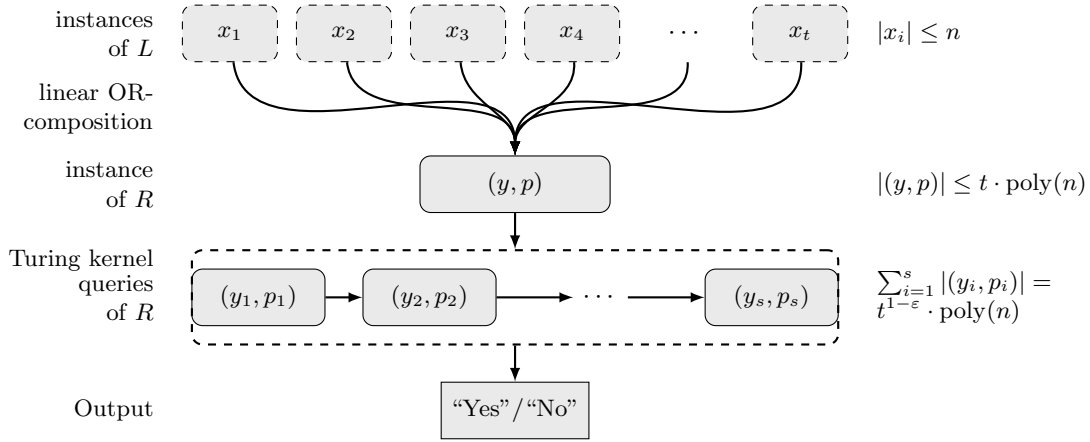

\subsection{Proof of Theorem~\ref{thm:lower-bound-main}}

The final ingredient of our lower bounds framework is showing that no \NP-hard problem admits an oracle OR distillation unless $\NP \subseteq \coNP/\poly$, proved in \Cref{thm:oracle OR distillation} below. Combining this theorem with \Cref{thm:CompositionTuringKernel} directly gives us \Cref{thm:lower-bound-main}. The proof of \Cref{thm:oracle OR distillation} is similar to the \citet{FortnowS11} argument used in original framework for (Karp) kernels~\cite{BodlaenderDFH09}; below we follow the exposition given in \cite[Proof of Theorem 15.3]{CyganFKLMPPS15}. 

The basic intuition is that even with total query size $t^{1-\vareps} \poly (n)$, the assumed oracle OR distillation of an \NP-hard problem $L$ can reduce any $t$ instances of $L$ in polynomial time to oracle queries of some parameterized problem $R$ that have total size smaller than~$t$. By basically using the same arguments as in~\cite{FortnowS11}, we can then construct a polynomial advice for $L$, arriving at the conclusion that $L \in \coNP / \poly$. For our purposes, we need to adapt the proof of~\cite{FortnowS11} at two places: First, to achieve the desired total size of the queries used by oracle OR distillation, we need select $t$ in a different manner (\emph{e.g.}, now it will also depend on~$\vareps$). Second, the advice used to show that $L \in \coNP / \poly$ will need to be more complex.

\begin{theorem}
\label{thm:oracle OR distillation}
Let $L$ be any problem, and let $R$ be a parameterized problem. If there exists an oracle OR distillation from $L$ into~$R$, then $L \in \coNP/\poly$.
\end{theorem}

\begin{proof}
Suppose that there exists an oracle OR distillation $\mathcal{D}$ from~$L$ into~$R$. Let $\vareps  > 0$ and $p$ be a polynomial such that $t^{1-\vareps} \cdot p(n)$ upper bounds the total size of the queries used by~$\mathcal{D}$, when $\mathcal{A}$ receives an input of $t$ instances of $L$, each of size at most~$n$. Without loss of generality, we assume that $p(n) \geq 2$ for all $n \in \N$, and that $2/\vareps \in \mathbb{N}$. For $n \in \mathbb{N}$, we set
\[ 
T(n) := p(n)^{2/\vareps} \text{ and } Q(n): = T(n). 
\]
Then when calling~$\mathcal{D}$ on $T(n)$ instances of $L$, each of size at most~$n$, 
the total size of the oracle queries is bounded by~$Q(n)/2$:
\[
T(n)^{1-\vareps} \cdot p(n) \,=\, \bigr( p(n)^{2/\vareps} \bigr)^{1-\vareps} \cdot p(n) 
\,=\, p(n)^{2/\vareps - 1} \,=\, T(n) / p(n) \,\le\,  Q(n)/2. 
\]

Now, consider the following function $f$ that maps $T(n)$ instances $x_1,\ldots,x_{T(n)} \in \{0,1\}^{\le n}$ of $L$ to a single string $a \in \{0,1\}^{\le Q(n)}$ according to the queries of $\mathcal{D}$ on $x_1,\ldots,x_{T(n)}$. Specifically, if $\mathcal{D}$ on input $x_1,\ldots,x_{T(n)}$ queries the instances $(y_1,p_1),\ldots,(y_q,p_q)$ of $R$, $q \leq Q(n)/2$, then we set 
\[
f(x_1,\ldots,x_{T(n)})  \mapsto (y_1,p_1)b_1(y_2,p_2)b_2\cdots (y_q,p_q)b_q,
\]
where $b_i$ is a Boolean value indicating whether or not~$y_i\in R$ for every $i \in \{1, \dots, q\}$. Note that since $\sum_i |(y_i,p_i)| \leq Q(n)/2$ and $q \leq Q(n)/2$, we have that $|f(x_1,\ldots,x_{T(n)})| \leq Q(n)$.

For $n \in \N$, define the following three sets: 
\[ 
L_n := L \cap \{0,1\}^{\le n}, \quad \xbar L_n := \{0,1\}^{\le n} \setminus L_n, \quad \text{and} \quad F_n:= f\left((\xbar L_n)^{T(n)}\right).
\]
The sets~$L_n$ and~$\xbar L_n$ consist of the \yes- and \no-instances of~$L$ of size at most~$n$, respectively. The set $F_n$ contains exactly those elements~$(y_1,p_1)b_1\cdots (y_q,p_q)b_q\in \{0,1\}^{\le Q(n)}$ which correspond to sequences of instances~$(y_1,p_1) \dots, (y_q,p_q)$ of~$R$ that algorithm~$\mathcal{D}$ queries given an input of $T(n)$ \no-instances of~$L$. We say that a set $A \subseteq F_n$ \emph{covers} $\xbar L_n$ if for every~$x\in \xbar L_n$ there exists some~$a\in A$ and $x_1, \dots, x_{T(n)} \in \{0,1\}^{\le n}$ such that $x = x_i$ for some $i \in \{1, \dots, T(n)\}$ and $f(x_1, \dots, x_{T(n)}) = a$. We will need the following claim:
\begin{claim}
\label{claim:coverset}
There exists a set $A_n \subseteq F_n$ that covers $\xbar L_n$ with $|A_n| \le n$.
\end{claim}

\begin{proof}[Proof of claim]
We consecutively construct strings $a_1, a_2, a_3\dots \in F_n$ until we cover $\xbar L_n$. Let $A_i=\{a_1,\ldots,a_i\}$, with $A_0=\emptyset$, and let $X_i$ be the set of strings from $\xbar L_n$ which are not covered by~$A_i$. Then $X_0 = \xbar L_n$. During the construction we will guarantee that $|X_i| \le 2^{-i} |\xbar L_n|$, which is of course satisfied at the beginning for $X_0$. Since $|\xbar L_n| \le |\{0,1\}^{\le n}| < 2^{n+1}$, it will then follow that $X_{n+1}$ is empty, and so the set $A_n=\{a_1, \ldots, a_n\}$ will cover $\xbar L_n$.
    
It remains to show how to construct $a_i$ given $A_{i-1}$. Recall that algorithm $\mathcal{D}$ maps every tuple from $(X_{i-1})^{T(n)}\subseteq (\xbar L_n)^{T(n)}$ into $F_n$. For a string $a \in F_n$, let $\mathcal{D}_{i-1}^{-1}(a) \subseteq (X_{i-1})^{T(n)}$ denote the set of all tuples $(x_1,\ldots,x_{T(n)}) \in (X_{i-1})^{(T(n)}$ with $\mathcal{D}(x_1,\ldots,x_{T(n)})=a$. Since $|F_n| \le 2^{Q(n)}$, it follows by the pigeon-hole principle that there exists some $a\in  F_n$ such that 
\[
|\mathcal{D}_{i-1}^{-1}(a)|  \,\ge\, \frac{|(X_{i-1})^{T(n)}|}{|F_n|} \,\ge\, \frac{|(X_{i-1})^{T(n)}|}{2^{Q(n)}} \, = \, \left(\frac{|X_{i-1}|}{2}\right)^{T(n)}.
\]

We choose $a_i$ to be a string $a \in F_n$ that satisfies the above equality. Observe that every string $x \in X_{i-1}$ that is contained in some $T(n)$-tuple of $\mathcal{D}_{i-1}^{-1}(a)$ is covered by the string~$a$. Consequently, $\mathcal{D}_{i-1}^{-1}(a) \subseteq (X_{i-1} \setminus X_i)^{T(n)}$, and so
\[
\bigl( \frac{|X_{i-1}|}{2}\bigr)^{T(n)} \le |\mathcal{D}_{i-1}^{-1}(a)| \le |(X_{i-1} \setminus X_i)^{T(n)}| = |X_{i-1} \setminus X_i|^{T(n)}.
\]
We infer that $|X_{i-1}\setminus X_i |\ge |X_{i-1}|/2$, and the condition $|X_i | \le 2^{-i} |\xbar L_n|$ follows by a trivial induction. As we have argued, this means that the construction terminates after $n+1$ steps, yielding a feasible set $A_n$. This concludes the proof of the claim.
\end{proof}
  
We now argue that $L \in \coNP/\poly$ by showing that its complement $\xbar L := \{0,1\}^* \setminus L$ is in $\NP/\poly$. That is, we show that a nondeterministic polynomial-time algorithm with access to polynomial-size advice can decide $\xbar L$. The advice our algorithm uses for an input of length $n$ is the set $A_n$ promised by~\cref{claim:coverset}. Note that as $A_n$ contains at most $n$ strings, each of length at most $Q(n)=\poly(n)$, the size of the advice is indeed polynomial in~$n$.
  
Let $x \in \{0,1\}^n$ for some $n \in \N$. Our nondeterministic algorithm~$\mathcal{A}$ decides whether $x \in \xbar L$ using of the advice $A_n$ as follows:
\begin{enumerate}
\item Guess nondeterministicly a sequence $x_1, \ldots, x_{T(n)} \in \{0,1\}^{\leq n}$. If there is no $i \in \{1, \ldots, T(n)\}$ such that $x =x_i$, then determine that $x \notin \xbar L$.
\item Simulate algorithm~$\mathcal{D}$ on input~$(x_1, \dots, x_{T(n)})$, once per each string in $A_n$ (so $n+1$ times in total). For each string $(y_1,p_1)b_1\cdots (y_q,p_q)b_q\in A_n$, in the simulation of $\mathcal{D}$ using this string, replace any oracle call $y$ that $\mathcal{D}$ makes by checking whether $y=y_i$ for some $\{1,\ldots,q\}$. If so, use the value $b_i$ to determine how to proceed with the simulation $\mathcal{D}$. Otherwise, if $y \neq y_i$ for all $i \in \{1,\ldots,q\}$, determine that $y \notin \xbar L$.
\item Conclude that~$x \in \xbar L$ if the simulation of~$\mathcal{D} $ returns \no, and otherwise conclude $x\notin \xbar L$.
\end{enumerate}
  
Suppose that $x \in \xbar L_n$. Since~$A_n$ covers~$\xbar L_n$, there exists a guess $x_1, \dots, x_{T(n)}\in \{0,1\}^{\leq n}$ with~$x_i=x$ for some~$i\in\set{T(n)}$ such that $f(x_1, \dots, x_{T(n)})\in A_n$. Note that for each instance of~$R$ that $\mathcal{D}$ queries given~$x_1, \ldots, x_{T(n)}$, the string $f(x_1,\ldots,x_{T(n)})$ contains the information whether this instance is contained in~$R$ or not. Thus, for this guess, algorithm~$\mathcal{A}$ correctly concludes that $x \in \xbar L_n$. Conversely, suppose that $x\notin \xbar L_n$. By the definition of a oracle OR distillation, every tuple $(x_1, \dots, x_{T(n)})$ with $x = x_i$ for some $i \in \{1, \dots, t\}$ is mapped to a string excluded from~$F_n\supseteq A_n$. Thus, there is no guess such that $\mathcal{A}$ concludes that~$x \in \xbar L_n$. The correctness of $\mathcal{A}$ thus follows. Moreover, the running time of $\mathcal{A}$ is indeed polynomial, since $T(n)$, $|A_n|$, and the running time of~$\mathcal{D}$ are all polynomial in $n$.
\end{proof}

\subsection{Proof of \Cref{thm:LowerBound}}

We now apply \Cref{thm:lower-bound-main} to several parameterizations of \textsc{Clique}.
For the parameter vertex cover number, we observe that the OR composition given by Bodlaender et al.~\cite{BodlaenderJK14} is indeed a linear one (technically, the OR composition given by Bodlaender et al.~\cite{BodlaenderJK14} requires all input instances to have the same number of vertices and same $k$, but this restriction can easily be achieved by adding apex or isolated vertices in a preprocessing step). Herein, they construct a graph with a vertex set consisting of the union of a set~$A$ of~$t$ vertices and a set~$B$ of~$kn'+3\binom{n'}{2}$ vertices, where~$n'$ denotes the maximum number of vertices from any input instance's graph. Vertices from~$A$ are only adjacent with vertices from~$B$, and thus the size of the constructed graph is in~$t \cdot \log (t)\cdot \poly(kn')$.
Thus, combining this with \Cref{thm:lower-bound-main} we immediately obtain:
\begin{corollary}
\label{cor:vc}
\textsc{Clique} parameterized by vertex cover number does not admit a polynomially sized Turing kernel with $O(n^{1 - \vareps })$ queries, where $n$ is the size of the input graph.
\end{corollary}

Note that the lower bound in~\cref{cor:vc} also holds for all parameters listed in~\cref{thm:UpperBound} except for bandwidth, cutwidth, and maximum degree,
since all these parameters are bounded by a polynomial in the vertex cover number of the graph (see also~\cref{fig:paramhier}). For the remaining parameters, we observe that the disjoint union yields a linear OR composition:
\begin{corollary}
\label{cor:params}
\textsc{Clique} parameterized by bandwidth, cutwidth, or maximum degree, does not admit a polynomially sized Turing kernel with $O(n^{1 - \vareps }) $ queries, where $n$ is the size of the input graph.
\end{corollary}

\begin{proof}
We can assume without loss of generality that all input instances search for a clique of the same size (if not, then we add apex vertices to all instances asking for smaller cliques). Taking the disjoint union is a linear OR composition. The statement then follows from \Cref{thm:lower-bound-main}.
\end{proof}

Combining \Cref{cor:vc} with \Cref{cor:params}, together with the observation that all remaining parameters are polynomial bounded by at least one of the parameters listed in these corollaries, completes the proof of \Cref{thm:LowerBound}.

\section{Upper Bounds}
\label{sec:upper-bounds}

In this section, we complement our lower bounds of the previous section by presenting polynomial-size Turing kernels (in fact, OR kernels) for \prob{Clique} that use a sublinear number of queries. In particular, we provide a complete proof of \cref{thm:UpperBound}.

We begin by describing a general ``query batching" technique for obtaining OR-kernels with sublinear number of queries that is especially suitable for \prob{Clique}. The first component of this technique is the following notion of \emph{trivial OR compositions}.
\begin{definition}
Let~$L$ be a parameterized problem. A \emph{trivial OR composition} for~$L$ is an algorithm that takes~$t$ instances~$(\inst_1,p_1),\dots,(\inst_t,p_t)$ of~$L$ and computes in time~$\poly(\sum_i (|\inst_i|+p_i))$ an instance~$(\inst,p)$ of~$L$ such that
$(\inst,p)\in L$ if and only if~$(\inst_i,p_i)\in L$ for some~$i\in\set{t}$.
\end{definition}

For all parameters we consider in this section, \prob{Clique} has a trivial OR composition by taking the disjoint union of input graphs (by adding apex vertices, we may assume that all instances search for a clique of the same size). This is formalized as follows.
 
\begin{observation}
\label{obs:cliquecompose}
Let~$G_1,\ldots,G_t$ be~$t$ graphs, and let~$G$ be the disjoint union of~$G_1,\ldots,G_t$.
Then,
$G$ has a $k$-clique if and only if~$G_i$ has a $k$-clique for some~$i\in\set{t}$, for any $k\in\N$.
\end{observation}

The following lemma encapsulates the connection between problems with trivial OR compositions and OR-kernels with sublinear number of queries.

\begin{lemma}
\label{lem:logshave}
Let~$L$ be some parameterized language with parameter~$p$ such that
\begin{enumerate}[(i)]
\item $L$ admits a trivial OR composition, and
\item $L$ is solvable in $2^{\poly(p)} \cdot \poly(n)$. 
\end{enumerate}
Then,
on any input instance~$(I,p)$ of encoding length~$n$,
if $L$ admits a polynomial-size \ORK{} with~$q(I)\in\N$ many queries,
then $L$ also admits a polynomial-size \ORK{} with~$\lceil q(I)/\log^c(n) \rceil$ queries, 
for any $c \in \mathbb{N}$.
\end{lemma}

\begin{proof}
Suppose $L$ admits a a polynomial-size \ORK{} with~$q(I)$ queries
for any input instance~$(I,p)$ of encoding length~$n$. 
We show how to construct a a polynomial-size \ORK{} with~~$\lceil q(I)/\log^c(n) \rceil$ many queries, for any $c \in \mathbb{N}$. 
Let~$ d\in \N $ such that the OR kernel queries instances of size at most~$O(p^d)$.
Observe that if $p \leq \log^{\frac{1}{d}} n$, then one can solve $(\inst,p)$ in $2^{O (\log n)} \cdot \poly(n) = \poly(n)$ time, and we are done. So assume  $p > \log^{\frac{1}{d}}(n)$. 

Apply the given \ORK{} on $(\inst,p)$ to obtain~$q\ceq q(I)$ queries, 
denoted by~$(\inst_1,p_1), \ldots, (\inst_{q},p_{q})$, 
where~$|\inst_i|+p_i\leq O(p^d)$ for every~$i\in\set{q}$ and some $d \in\mathbb{N}$.
Form a partition of size~$r\ceq r(I)\leq \ceil{q(I)/\log^c (n)}$ on the output instances, such that each part contains at most~$\log^c (n)$ output instances.
Employ the trivial OR composition to each part to obtain the instances~$(J_1,p_1),\dots,(J_{r},p_{r})$, where~$|J_i|+p_i\leq \poly(p^d \log^c (n)) = \poly (p)$ for every~$i\in\set{r}$. 
Note that~$(\inst,p)$ is a \yes-instance if and only if $(J_i,p_i)$ is a \yes-instance for some~$i\in\set{r}$.
\end{proof}

We note that the lemma also applies to languages solvable in $f(p) \cdot \poly(n)$, 
for $f(p) = \omega(2^{\poly(p)})$, 
but this would result in a larger number of queries 
(albeit still asymptotically less then $q(I)$). 
Nevertheless, \cref{lem:logshave} as stated above suffices for our purposes. 

Note that \prob{Clique} is solvable in $2^{\poly(p)} \cdot \poly(n)$ time, for all parameters $p$ mentioned in \cref{thm:UpperBound}. Thus, together with \cref{obs:cliquecompose}, we have that \prob{Clique} satisfies both conditions of the \cref{lem:logshave} for all parameters mentioned in \cref{thm:UpperBound}. Below we show that \prob{Clique} admits polynomial-sized OR-kernels with a linear number of queries for each of these parameters, thus completing the proof of \cref{thm:UpperBound}.

\subsection{Degeneracy}

An undirected graph~$G$ has degeneracy~$p$ if every subgraph of~$G$ contains a vertex of degree at most~$p$. The degeneracy ordering of the vertices of a graph~$G$ with degeneracy~$p$ is an ordering of the vertices where every vertex has at most~$\dgn$ neighbors of higher index. The degeneracy ordering can be computed in linear time~\cite{MB83}. We prove the following.

\begin{lemma}
\label{lem:degen}
\Clique{} admits an~$O(\dgn^2)$\hyp size \ORK{} with at most~$n_G$ queries, where~$p$ denotes the degeneracy of the input graph~$G$.
\end{lemma}

\begin{proof}
Let~$(G,k)$ be an instance of~\prob{Clique} where~$G$ contains~$n_G$ vertices.
Compute the degeneracy~$\dgn$ of~$G$ degeneracy order~$\sigma$ in linear time. For each vertex~$v\in V(G)$, query the instance~$(G_v,k)$, where~$G_v$ is the subgraph of~$G$ induced by the closed-neighborhood of~$v$ restricted to vertices of higher index. Clearly, $G$ has a $k$-clique if and only if $G_v$ has a $k$-clique for some~$v \in V(G)$. Moreover, each $G_v$ has $O(p)$ vertices and $O(p^2)$ edges.
\end{proof}

Note that \prob{Clique} is solvable in $2^{O(p)} n^{\O(1)}$ time, by exhaustively searching each of the subgraphs $G_v$ defined in the proof of \cref{lem:degen}. Thus, by \cref{lem:logshave}, we have the following.

\begin{corollary}
For every~$c \in \mathbb{N}$, \Clique{} admits a polynomial-size \ORK{} with at most $n_G/\tlog{n_G}[c]$ output instances, where~$p$ denotes the degeneracy of the input graph~$G$.
\end{corollary}

\begin{remark}
  For graphs with $\omega (n_G^{1 + \delta})$ for some constant~$\delta > 0$ many edges, this \ORK{} uses truly sublinear many queries.
  This is no contradiction to \cref{thm:LowerBound} because \cref{thm:LowerBound} only excludes Turing kernels using a sublinear number of queries on \emph{every} input (in particular, looking in the proof of \cref{thm:LowerBound} for degeneracy and some fixed~$\epsilon > 0$, one sees that the assumed Turing kernel is applied to an instance~$(G, k)$ with~$O(n_G^{1 + \frac{\epsilon}{2}})$ many edges). 
\end{remark}

\subsection{Odd cycle transversal}

The odd cycle transversal number of a graph~$G$ is size $p$ of the smallest set~$X \subseteq V(G)$ such that~$G-X$ is bipartite.
The respective \prob{Odd Cycle Transversal} problem is solvable in~$\O(3^p \cdot n_G \cdot m_G)$ time~\cite{ReedSV04,Huffner09}, and hence it is polynomial-time solvable if~$p = O(\tlog{n_G})$. Moreover, the problem admits a polynomial-time $O(\sqrt{\tlog{n_G}})$-factor approximation~\cite{DBLP:conf/stoc/AgarwalCMM05} (see also \cite[Lemma 3.3]{KW14}).
It follows that it admits an polynomial-time $O(\sqrt{p})$-factor approximation in the case of~$p=\Omega(\tlog{n_G})$~\cite[Lemma 3.3]{KW14}.
\begin{lemma}
\label{lem:OCT}
\Clique{} admits a~$\O(\oct^3)$\hyp size \ORK{} with at most~$m_G$ output instances,
where~$\oct$ denotes the odd cycle transversal number of the input graph~$G$.
\end{lemma}

\begin{proof}
Let~$(G,k)$ be an input instance of~\Clique{} with~$m_G$ edges, and let $p$ be the odd cycle traversal number of $G$.
If $p \leq \tlog{n_G}$ then we can determine whether $G$ has a $k$-clique in polynomial-time.
So assume $p> \tlog{n_G}$.
Compute in polynomial-time an~$O(\sqrt{p})$-factor approximate minimum odd cycle transversal, and let $X$ denote the resulting solution of size at most~$O(p^{1.5})$. Let~$G-X$ consist of the edges~$e_1,\dots,e_{m'}$, where~$m'\leq m_G$.
For each~$i\in\set{m'}$, compute the query~$G[e\cup X]$. Clearly, since $G-X$ is bipartite, $G$ has a $k$-clique if and only if $G[e\cup X]$ has a $k$-clique for some $e \in \{e_1,\dots,e_{m'}\}$. Moreover, each query $G[e\cup X]$ has $O(p^{1.5})$ vertices and $O(p^3)$ edges.
\end{proof}

Observe that \prob{Clique} is solvable in $2^{O(p)} \poly (n)$ time by exhaustively searching through each subgraph $G[e\cup X]$ defined in the proof of \cref{lem:OCT}.
By \cref{lem:logshave}, we therefore have the following:

\begin{corollary}
For every~$c\geq0$, \Clique{} admits a polynomial-size \ORK{} with at most $m_G/\tlog{m_G}[c]$ queries, where~$p$ denotes the odd cycle transversal number of the input graph~$G$.
\end{corollary}

\subsection{Distance to bounded degree}

For an undirected graph~$G$, the distance of $G$ to graphs of bounded degree~$d$ of $G$ is the smallest number~$p$ such that there is a set~$X\subseteq V(G)$ of size~$p$ with~$G-X$ having maximum degree~$d$.
We begin by presenting an~$O(p)$-factor approximation algorithm for finding the smallest such set~$X \subseteq V(G)$.

\begin{lemma}
\label{lem:approxdbd}
Let~$d\in\N$. There is an algorithm that, given an undirected graph~$G$ and an integer~$p\in\N$, in polynomial-time either returns a set~$X$ of size at most~$p\cdot (p+d+1)$ such that~$G-X$ has maximum degree at most~$d$, or correctly concludes that there is no such set of size at most~$p$. 
\end{lemma}

\begin{proof}
Every vertex of degree at least~$p+d+1$ in $G$ must be contained in any solution. Thus, if there are more than~$p$ such vertices, we can correctly report \no{} (\emph{i.e.}, that the distance to bounded degree $d$ is larger than~$p$). Otherwise, we know that the maximum degree is at most~$p+d$ in the subgraph~$G'$ obtained by deleting (up to~$p$) vertices of degree at least~$p +d + 1$ from $G$. If there are more than~$p\cdot (p+d)$ vertices of degree larger than~$d$ remaining in~$G'$, then the distance of $G$ must larger than~$p$, and we can correct report \no. Thus, there are at most~$p\cdot (p+d)$ vertices of degree larger than~$d$, and the set of these vertices is a $(p+d+1)$-approximation.
\end{proof}

\begin{lemma}
For any~$d\in\N$, \Clique{} admits an~$O(p^4)$-size \ORK{} with at most $n_G$ queries, where~$p$ denotes the distance of the input graph~$G$ to graphs of maximum degree~$d$.
\end{lemma}

\begin{proof}
Let~$(G,k)$ be an instance of~\prob{Clique}, where~the distance to bounded degree~$d$ of~$G$ is~$p$. Iteratively, for~$p'\in\{1,\dots,p\}$, we apply~\cref{lem:approxdbd} to obtain a set~$X$ with~$|X|\leq p\cdot (p+d+1)$ such that~$H=G-X$ is of maximum degree~$d$. Next, for every~$v\in V(H)$, we query the instance~$(G_v,k)$, where~$G_v=G[X\cup N_H[v]]$. Clearly, $G$ has a $k$-clique if and only if $G$ has a $k$-clique for some $v \in V(H)$. Moreover,  observe that any~$G_v$ has at most~$p\cdot (p+d+1)+d+1=O(p^2)$ vertices, and at most $O(p^4)$ edges.
Thus, in total we query at most~$n_G$ instances each of size~$\O(p^4)$.
\end{proof}

\begin{corollary}
For any~$c,d\in\N$, \Clique{} admits a polynomial-size \ORK{} with at most $n_G/\tlog{n_G}[c]$ queries, where~$p$ denotes the distance of the input graph~$G$ to graphs of maximum degree~$d$.
\end{corollary}

\subsection{Distance to chordal graph}

For an undirected graph~$G$, the distance of $G$ to chordal is the smallest number~$p$
such that there is a set~$X\subseteq V(G)$ of size~$p$ with~$G-X$ being chordal; that is,
every induced cycle in~$G-X$ is of length exactly three. We will make use of the following.

\begin{lemma}[\cite{JansenP18}]
\label{lem:chordal}
There is a polynomial-time algorithm that, given an undirected graph~$G$ and an integer~$p$,
either correctly concludes that~$G$ has distance to chordal graphs greater than~$p$, or computes a set $X$ of size~$q=O(p^4 \tlog{p}[2])$ such that $G-X$ is chordal. 
\end{lemma}

\begin{lemma}
\Clique{} admits an~$O((q+k)^2)$-size \ORK{} with at most $n_G$ queries, where~$q=O(p^4\tlog{p}[2])$ and~$p$ denotes the distance of the input graph~$G$ to chordal graphs.
\end{lemma}

\begin{proof}
Let~$(G,k)$ be an instance of~\prob{Clique}, and let~$p$ denote the distance of $G$ to chordal graphs. Apply~\cref{lem:chordal} to obtain a set~$X$ of size~$q\in\O(p^4 \tlog{p}[2])$ such that $H= G-X$ is chordal. We compute a perfect elimination ordering~$\sigma$ of~$H$ in polynomial time~\cite{HabibMPV00}. That is, an ordering of the vertices of $H$ such that for each vertex along with its higher indexed neighbors induce a clique. Let~$N_\sigma^+(v)$ denote the neighborhood of~$v$ restricted to vertices which are ordered after~$v$ in~$\sigma$, and let~$N^+[v]\ceq N_\sigma^+(v)\cup\{v\}$. We query each instance $(G_v,k)$, where~$G_v=G[X\cup N_\sigma^+[v]]$. Clearly, $G$ has a $k$-clique if and only if $G_v$ has a $k$-clique for some $v \in V(H)$. Moreover, each query contains at most~$|X|+k$ many vertices.
In total, there are at most~$n_G$ output instances each of size at most~$(|X|+k)^2$.
\end{proof}

\begin{corollary}
For any~$c\in\N$, \Clique{} admits a polynomial-size \ORK{} with $n_G/\tlog{n_G}[c]$ queries, 
where~$p$ denotes the distance of the input graph~$G$ to chordal graphs.
\end{corollary}

\subsection{Length of a longest odd cycle}

We next show a $p^{O(1)}$-size \ORK{} for \prob{Clique} with a linear number of queries, where $p$ is the length of the longest odd cycle in the input graph $G$. For this, we will need the notion of tree decompositions.
\begin{definition}[Tree Decomposition]
A \emph{tree decomposition} of a graph $G$ is a tree $\mathcal{T}=(\mathcal{X},F)$ whose nodes $X \in \mathcal{X}$ (called \emph{bags}) are subsets of vertices of $G$ satisfying:
\begin{inparaenum}[(i)]
\item $G= \bigcup_{X \in \mathcal{X}} G[X]$,
\item for every edge~$e= \{v, w\} \in E(G)$, we have $\{v, w\} \subseteq X$ for some~$X\in \mathcal{X}$, and
\item the set of bags $\mathcal{X}_v = \{ X \in \mathcal{X} : v \in X \}$, for every $v \in V(G)$, induces a connected subtree in~$\mathcal{T}$. 
\end{inparaenum}
The \emph{width} of the decomposition is $\max_{x \in \mathcal{X}} |X| - 1$. The \emph{treewidth} of $G$ is the minimum width of any tree decomposition of $G$. 
\end{definition}

Fomin~\emph{et al.}~\cite{FominLSPW18} showed how to compute in polynomial-time a tree decomposition of width $O(\omega^2)$ of a graph of treewidth~$\omega$. Let $G$ be a graph, and $\mathcal{T}$ be a tree decomposition of~$G$. We will use the following two well known facts about $\mathcal{T}$ (see \emph{e.g.}~\cite{FlumG06}): First, any clique of~$G$ is completely contained in some bag of $\mathcal{T}$. Second, in linear time, we can modify $\mathcal{T}$ so it has at most~$|V(G)|$ many bags.
Note that the first fact implies a $2^{\omega} \poly (n_G)$ time algorithm for \prob{Clique}, where $\omega$ is the width of $\mathcal{T}$, by exhaustively searching each of the bags of $\mathcal{T}$.
A \emph{cutvertex} in a graph is a vertex whose removal leaves the graph with more connected components. A connected graph with at least three vertices and no cutvertex is called \emph{2-connected}.
A \emph{block} of a graph~$G$ is a maximal 2-connected subgraph $G$.
A \emph{block decomposition} of a graph is the set of all its blocks.
Clearly, there are~$O(n_G)$ many blocks in a block decomposition, and any non-trivial clique of~$G$ is completely contained inside one of its blocks. We using the following result by Panolan and Rai~\cite{PanolanR12}:

\begin{lemma}[\cite{PanolanR12}]
\label{lem:loc}
Let~$G$ be a graph with no odd cycle of length greater than~$p$. Then, in~$O(n_G^2)$ time, one can compute the block decomposition of $G$, and each block of~$G$ is either
\begin{inparaenum}[(i)]
\item bipartite, or 
\item non-bipartite, 2-connected, and of treewidth at most~$2p$.
\end{inparaenum}    
\end{lemma}

Let $B=\{B_1,\dots,B_b\}$ be the block decomposition of a graph $G$, and let $C(V)$ denote the set of cutvertices of $G$. Observe that any vertex of $V(G) \setminus C(V)$ appears in exactly one block. Moreover, any cutvertex $v$ appears in at most~$\deg(v)$ many blocks. Thus, we have 
$
\sum_{i=1}^b |B_i| \leq n_G + \sum_{v \in C(V)} \deg(v) \leq n_G +2m_G.
$
We will use this fact for our \ORK{}: 

\begin{lemma}
\label{lem:locork}
\Clique{} admits an~$O(p^4)$-size \ORK{} with at most $n_G +2m_G$ queries, where~$p$ denotes the length of the longest odd cycle in the input graph~$G$.
\end{lemma}

\begin{proof}
Let~$(G,k)$ be an instance of \prob{Clique}, and let~$p$ denote the length of the longest odd cycle in~$G$. We can assume w.l.o.g. that $k > 2$. Apply~\cref{lem:loc} to obtain a block decomposition~$B=\{B_1,\dots,B_b\}$, where~$B_i=(W_i,F_i)$ with~$W_i\subseteq V(G)$, $F_i\subseteq E(G)$, and~$\bigcup_{i\in\set{b}} B_i = G$. For each~$i\in\set{b}$, check in polynomial time whether~$B_i$ is bipartite, and if not, compute a tree decomposition $\mathcal{T}_i=(\mathcal{X}_i,F_i)$ of width~$O(p^2)$ for $B_i$ (using~\cite{FominLSPW18}). If necessary, in linear time, modify $\mathcal{T}_i$ so that it consists of at most~$|B_i|$ bags. Then, construct a query $G[X]$ for each bag $X \in \mathcal{X}_i$. By the above discussion, any $k$-clique of $G$ is completely contained in some bag of some non-bipartite block of $G$. Thus, $G$ has a $k$-clique if and only if $G[X]$ for some $X \in \mathcal{X}_i$ and some $i \in \set{b}$.
Moreover, the number of queries is bounded by $\sum_i |B_i|$, which according to the above is bounded by $n_G + 2m_G$.
Thus, in total, we have at most~$n_G +2m_G$ queries, each of size~$O(p^4)$.
\end{proof}

\begin{corollary}
For any~$c\in\N$, \Clique{} admits a polynomial-size \ORK{} with $m_G/\tlog{m_G}[c]$ queries, where~$p$ denotes the denote the length of the longest odd cycle in the input graph~$G$.
\end{corollary}

\section{Conclusion}

In this paper we showed that \Clique{} under several parameterizations admits polynomial Turing kernels (in fact, OR-kernels) with $O(n/\tlog{n}[c])$ queries, yet it does not admit such kernels with $O(n^{1-\varepsilon})$ queries. For our upper bounds, we used several known kernelization techniques for \Clique{} along with a new simple query batching technique. To obtain our lower bounds, we adapted the standard Karp kernelization lower bounds framework~\cite{BodlaenderDFH09}, by employing a few new tricks and ideas. The natural future direction of research is to see if our framework can be applied to further problems, perhaps with a few extra adjustments. Moreover, it will be interesting to identify natural problems without polynomial Karp kernels that have polynomial Turing kernels using only a few number of queries.

\bibliographystyle{abbrvnat}
\bibliography{tk-queries-bib}

\end{document}